\documentclass[11pt]{article}%
\usepackage{amsmath}
\usepackage{amsfonts}
\usepackage{amssymb}
\usepackage{graphicx}
\usepackage{color}%
\providecommand{\U}[1]{\protect\rule{.1in}{.1in}}
\newtheorem{theorem}{Theorem}[section]

\newtheorem{corollary}{Corollary}[section]

\newtheorem{lemma}{Lemma}[section]

\newtheorem{remark}{Remark}[section]

\newenvironment{proof}[1][Proof]{\noindent\textbf{#1.} }{\ \rule{0.5em}{0.5em}}
\numberwithin{equation}{section}

\begin{document}

\title{On the Frequency of Drawdowns for Brownian Motion Processes}
\author{David Landriault\thanks{Department of Statistics and Actuarial Science,
University of Waterloo, Waterloo, ON, N2L 3G1, Canada (dlandria@uwaterloo.ca)}
\and Bin Li\thanks{Corresponding Author: Department of Statistics and Actuarial
Science, University of Waterloo, Waterloo, ON, N2L 3G1, Canada
(bin.li@uwaterloo.ca)}
\and Hongzhong Zhang\thanks{Department of Statistics, Columbia University, New
York, NY, 10027, USA (hzhang@stat.columbia.edu)} }
\date{{\small \today}}
\maketitle

\begin{abstract}
Drawdowns measuring the decline in value from {the} historical running maxima
over a given period of time, are considered as extremal events from the
standpoint of risk management. To date, research on the topic has mainly focus
on the side of severity by studying the first drawdown over certain
{pre-specified} size. In this paper, we extend the discussion by investigating
the frequency of drawdowns, and some of their inherent characteristics. We
{consider} two types of drawdown time sequences depending on whether a
historical running maximum {is reset or not}. For each type, we study the
frequency rate of drawdowns, the Laplace transform of the $n$-th drawdown
time, the distribution of the running maximum and the value process at the
$n$-th drawdown time, as well as some other quantities of interest.
Interesting relationships between these two drawdown time sequences are also
established. Finally, insurance policies protecting against the risk of
frequent drawdowns are also proposed and priced.



\end{abstract}

\baselineskip15.5pt

\section{Introduction}

We consider a drifted Brownian motion $X=\{X_{t},t\geq0\}$, defined on a
filtered probability space $(\Omega,\{\mathcal{F}_{t},t\geq0\},\mathbb{P})$,
with dynamics%
\[
X_{t}=x_{0}+\mu t+\sigma W_{t},
\]
where $x_{0}\in%
\mathbb{R}
$ is the initial value, $\mu\in%
\mathbb{R}
$, $\sigma>0$, and $\{W_{t},t\geq0\}$ is a standard Brownian motion. The time
of the first drawdown over size $a>0$ is denoted by
\begin{equation}
\tau_{a}:=\inf\{t>0:M_{t}-X_{t}\geq a\}, \label{tau a}%
\end{equation}
where $M=\left\{  M_{t},t\geq0\right\}  $ with $M_{t}:=\sup_{s\in\lbrack
0,t]}X_{t}$ is the running maximum process of $X$. Here and henceforth, we
follow the convention that $\inf{\emptyset}=\infty$ and $\sup{\emptyset}=0$.

Drawdown is one of the most frequently quoted path-dependent risk indicators
for mutual funds and commodity trading advisers (see, e.g., Burghardt et al.
\cite{DrawdownRisk}). From a risk management standpoint, large drawdowns
should be considered as extreme events of which both the severity and the
frequency need to be investigated. Considerable attention has been paid to the
severity aspect of the problem by {pre-specifying} a threshold, namely $a>0$,
of the size of drawdowns, and subsequently studying various properties
associated to the first drawdown time $\tau_{a}$. In this paper, we extend the
discussion by investigating the frequency of drawdowns. To this end, we derive
the joint distribution of the $n$-th drawdown time, the running maximum, and
the value process at the drawdown time for a drifted Brownian motion. Using
the general theory on renewal process, we proceed to characterize the behavior
of the frequency of drawdown episodes in a long time-horizon. Finally, we
introduce some insurance policies which protect against the risk associated
with frequent drawdowns. These policies are similar to the sequential barrier
options in over-the-counter (OTC) market (see, e.g., Pfeffer \cite{Pfef01}).
Through Carr's randomization of maturities, we provide closed-form pricing
formulas by making use of the main theoretical results of the paper.

\subsection{Literature review}

The first drawdown time $\tau_{a}$ is the first passage time of the drawdown
process $\left\{  M_{t}-X_{t},t\geq0\right\}  $ to level $a$ or above. It has
been extensively studied in the literature of applied probability. The joint
Laplace transform of $\tau_{a}$ and $M_{\tau_{a}}$ was first derived by Taylor
\cite{Taylor75} for a drifted Brownian motion. Lehoczky \cite{Lehoczky77}
extended the results to a general time-homogeneous diffusion by a perturbation
approximation approach. An infinite series expansion of the distribution of
$\tau_{a}$ was derived by Douady et al. \cite{DSY00} for a standard Brownian
motion and the results were generalized to a drifted Brownian motion by Magdon
et al. \cite{MDD04}. {The dual of drawdown, known as drawup, measures the
increase in value from the historical running minimum over a given period of
time. }The probability that a drawdown precedes a drawup is subsequently
studied by Hadjiliadis and Vecer \cite{HadjVece} and Pospisil et al.
\cite{PospVeceHadj} under the drifted Brownian motion and the general
time-homogeneous diffusion process, respectively. Mijatovic and Pistorius
\cite{MijatovicPistorius} derived the joint Laplace transform of $\tau_{a}$
and the last passage time at level $M_{\tau_{a}}$ prior to $\tau_{a}$,
associated with the joint distribution of the running maximum, the running
minimum, and the overshoot at $\tau_{a}$ for a spectrally negative L\'{e}vy
process. The probability that a drawdown precedes a drawup in a finite
time-horizon is studied under drifted Brownian motions and simple random walks
in \cite{ZhanHadj}. {More recently, \cite{Zhang13, ZhanHadj12} studied Laplace
transforms of the drawdown time, the so-called speed of market crash, and
various occupation times at the first exit and the drawdown time for a general
time-homogeneous diffusion process.}

In quantitative risk management, drawdowns and its descendants have become an
increasingly popular and relevant class of path-dependent risk indicators. A
portfolio optimization problem with constraints on drawdowns was explicitly
solved by Grossman and Zhou \cite{GrosZhou} in a Black-Scholes framework.
Hamelink and Hoesli \cite{HameHoes04} used the relative drawdown as a
performance measure in optimization of real estate portfolios. Chekhlov et al.
\cite{ChekUryaZaba} proposed a new family of risk measures called conditional
drawdown and studied parameter selection techniques and portfolio optimization
under constraints on conditional drawdown. Some {novel} financial derivatives
were introduced by Vecer \cite{Vece06} to hedge maximum drawdown risk.
Pospisil and Vecer \cite{PospVece} invented a class of Greeks to study the
sensitivity of investment portfolios to running maxima and drawdowns. Later,
Carr et al. \cite{CarrZhanHaji} introduced a class of European-style digital
drawdown insurances and proposed semi-static hedging strategies using barrier
options and vanilla options. The swap type insurances and cancelable insurances
against drawdowns were studied in Zhang et al. \cite{ZhanLeunHadj}.

\subsection{Definitions}

{While sustaining downside risk can be appropriately characterized using the
drawdown process and the first drawdown time, economic turmoil and volatile
market fluctuations are better described by quantities containing more
path-wise information, such as the frequency of drawdowns. }The existing
knowledge about the first drawdown time $\tau_{a}$ provides only limited and
implicit information about the frequency of drawdowns. For the purpose of
tackling the problem of frequency directly and systematically, we define below
two types of drawdown time sequences depending on whether the last running
maximum needs to be recovered or not.

The first sequence $\{\tilde{\tau}_{a}^{n},n\in%
\mathbb{N}
\}$ is called the \emph{drawdown times with recovery}, defined recursively as%
\begin{equation}
\tilde{\tau}_{a}^{n}:=\inf\{t>\tilde{\tau}_{a}^{n-1}:M_{t}-X_{t}\geq
a,M_{t}>M_{\tilde{\tau}_{a}^{n-1}}\}, \label{tau til}%
\end{equation}
where $\tilde{\tau}_{a}^{0}=0$. Note that, after each $\tilde{\tau}_{a}^{n-1}%
$, the corresponding running maximum $M_{\tilde{\tau}_{a}^{n-1}}$ must be
recovered before the next drawdown time $\tilde{\tau}_{a}^{n}$. In other
words, the running maximum is reset and updated only when the previous one is
revisited. Since the sample paths of $X$ are almost surely (a.s.) continuous,
we have that $M_{\tilde{\tau}_{a}^{n}}-X_{\tilde{\tau}_{a}^{n}}=a$ a.s. if
$\tilde{\tau}_{a}^{n}<\infty$.

The second sequence $\{\tau_{a}^{n},n\in%
\mathbb{N}
\}$ is called the \emph{drawdown times without recovery}, defined recursively
as%
\begin{equation}
\tau_{a}^{n}:=\inf\{t>\tau_{a}^{n-1}:M_{[\tau_{a}^{n-1},t]}-X_{t}\geq a\},
\label{tau}%
\end{equation}
where $\tau_{a}^{0}:=0$ and $M_{[s,t]}:=\sup_{s\leq u\leq t}X_{u}$. From
definition (\ref{tau}), it is implicitly assumed that the running maximum
$M_{\tau_{a}^{n}}$ is \textquotedblleft reset\textquotedblright\ to
$X_{\tau_{a}^{n}}$ at the drawdown time $\tau_{a}^{n}$. In fact, $\tau_{a}%
^{n}$ is the so-called iterated stopping times associated with $\tau_{a}$
defined as%
\begin{equation}
\tau_{a}^{n}=\left\{
\begin{array}
[c]{ll}%
\tau_{a}^{n-1}+\tau_{a}\circ\theta_{\tau_{a}^{n-1}}, & \text{when }\tau
_{a}^{n-1}\text{ and }\tau_{a}\circ\theta_{\tau_{a}^{n-1}}\text{ are
finite,}\\
\infty, & \text{otherwise,}%
\end{array}
\right.  \label{iterated}%
\end{equation}
where $\theta$ is the Markov shift operator such that $X_{t}\circ\theta
_{s}=X_{s+t}$ for $s,t\geq0$.

Note that both $\tau_{a}^{n}$ and $\tilde{\tau}_{a}^{n}$ are independent of
the initial value $x_{0}$ for not only the drifted Brownian motion $X$, but
also a general L\'{e}vy process. In view of definitions (\ref{tau}) and
(\ref{tau til}), it is clear that the following inclusive relation of the two
types of drawdown times holds:%
\[
\{\tilde{\tau}_{a}^{n},n\in%
\mathbb{N}
\}\subset\{\tau_{a}^{n},n\in%
\mathbb{N}
\}.
\]
In other words, for each $n\in%
\mathbb{N}
$, there exists a unique positive integer $m\geq n$ such that $\tilde{\tau
}_{a}^{n}=\tau_{a}^{m}$ (if $\tilde{\tau}_{a}^{n}<\infty$).

Our motivation for introducing the two drawdown time sequences are as follows.
The drawdown times with recovery $\{\tilde{\tau}_{a}^{n},n\in%
\mathbb{N}
\}$ are easy to identify from the sample paths of $X$ by searching the running
maxima. Moreover, they are consistent with definition (\ref{tau a}) of the
first drawdown $\tau_{a}$ in the sense that a drawdown can be considered as
incomplete if the running maximum has not been revisited. However, there are
also some crucial drawbacks of $\{\tilde{\tau}_{a}^{n},n\in%
\mathbb{N}
\}$ which motivate us to introduce the drawdown times without recovery
$\{\tau_{a}^{n},n\in%
\mathbb{N}
\}$. First, the downside risk during recovering periods is neglected. One or
more larger drawdowns may occur in a recovering period. Second, the threshold
$a$ needs to be adjusted to gain a more integrated understanding about the
severity of drawdowns. In other words, the selection of $a$ becomes tricky.
Third, the requirement of recovery is too strong. In real world, a historical
high water mark may never be recovered again, {as in the case of a financial
bubble \cite{Bubble11}.}

The rest of the paper is organized as follows. In Section 2, some
preliminaries on exit times and the first drawdown time $\tau_{a}$ of the
drifted Brownian motion $X$ are presented. In Section 3, the frequency rate of
drawdowns, and the Laplace transform of $\tilde{\tau}_{a}^{n}$ associated with
the distribution of $M_{\tilde{\tau}_{a}^{n}}$ and/or $X_{\tilde{\tau}_{a}%
^{n}}$ are derived. Section 4 is parallel to Section 3 but studies the
drawdown times without recovery $\{\tau_{a}^{n},n\in%
\mathbb{N}
\}$. Interesting connections between the two drawdown time sequences are
established. In Section 5, some insurance contracts are introduced to insure
against the risk of frequent drawdowns.

\section{Preliminaries}

Henceforth, for ease of notation, we write $\mathbb{E}_{x_{0}}[\,\cdot
\,]=\mathbb{E}[\left.  \cdot\,\right\vert X_{0}=x_{0}]$ for the conditional
expectation, $\mathbb{P}_{x_{0}}\{\,\cdot\,\}$ for the corresponding
probability and $\mathbb{E}_{x_{0}}[\,\cdot\,;U]=\mathbb{E}_{x_{0}}%
[\,\cdot\,{1}_{U}]$ with ${1}_{U}$ denoting the indicator function of a set
$U\subset\Omega$. In particular, when $x_{0}=0$, we drop the subscript $x_{0}$
from the conditional expectation and probability.

For $x\in%
\mathbb{R}
$, let $T_{x}^{+}=\inf\left\{  t\geq0:X_{t}>x\right\}  $ and$\ T_{x}^{-}%
=\inf\left\{  t\geq0:X_{t}<x\right\}  $ be the first passage times of $X$ to
levels in $\left[  x,\infty\right)  $ and $\left(  -\infty,x\right]  $,
respectively. For $a<x<b$ and $\lambda>0$, it is known that%
\begin{equation}
\mathbb{E}_{x}[\mathrm{e}^{-\lambda T_{a}^{-}}]=\mathrm{e}^{\beta_{\lambda
}^{-}(x-a)}\qquad\text{and}\qquad\mathbb{E}_{x}[\mathrm{e}^{-\lambda T_{b}%
^{+}}]=\mathrm{e}^{\beta_{\lambda}^{+}(x-b)}, \label{one-sided L}%
\end{equation}
where $\beta_{\lambda}^{\pm}=\frac{-\mu\pm\sqrt{\mu^{2}+2\lambda\sigma^{2}}%
}{\sigma^{2}}$ (see, e.g., formula 2.0.1 on Page 295 of Borodin and Salminen
\cite{borodin2002handbook}). By letting $\lambda\rightarrow0+$ in
(\ref{one-sided L}), we have%
\begin{equation}
\mathbb{P}_{x}\left\{  T_{b}^{+}<\infty\right\}  =\mathrm{e}^{\frac{-\mu
+|\mu|}{\sigma^{2}}(x-b)}\qquad\text{and}\qquad\mathbb{P}_{x}\left\{
T_{a}^{-}<\infty\right\}  =\mathrm{e}^{\frac{-\mu-|\mu|}{\sigma^{2}}(x-a)}.
\label{one-sided P}%
\end{equation}

From Taylor \cite{Taylor75} or Equation (17) of Lehoczky \cite{Lehoczky77}, we
have the following joint Laplace transform of the first drawdown time
$\tau_{a}$ and its running maximum $M_{\tau_{a}}$.

\begin{lemma}
\label{lem 1}For $\lambda,s>0$, we have%
\begin{equation}
\mathbb{E}\left[  \mathrm{e}^{-\lambda\tau_{a}-sM_{\tau_{a}}}\right]
=\frac{c_{\lambda}}{b_{\lambda}+s} \label{JL}%
\end{equation}
where $b_{\lambda}=\frac{\beta_{\lambda}^{+}\mathrm{e}^{-\beta_{\lambda}^{-}%
a}-\beta_{\lambda}^{-}\mathrm{e}^{-\beta_{\lambda}^{+}a}}{\mathrm{e}%
^{-\beta_{\lambda}^{-}a}-\mathrm{e}^{-\beta_{\lambda}^{+}a}}$ and $c_{\lambda
}=\frac{\beta_{\lambda}^{+}-\beta_{\lambda}^{-}}{\mathrm{e}^{-\beta_{\lambda
}^{-}a}-\mathrm{e}^{-\beta_{\lambda}^{+}a}}$.
\end{lemma}

A Laplace inversion of (\ref{JL}) with respect to $s$ results in
\begin{equation}
\mathbb{E}[\mathrm{e}^{-\lambda\tau_{a}};M_{\tau_{a}}>x]=\frac{c_{\lambda}%
}{b_{\lambda}}\mathrm{e}^{-b_{\lambda}x},\label{L1}%
\end{equation}
for $x>0$. Furthermore, letting $x\rightarrow0+$ in (\ref{L1}), we immediately
have
\begin{equation}
\mathbb{E}[\mathrm{e}^{-\lambda\tau_{a}}]=c_{\lambda}/b_{\lambda}.\label{lap}%
\end{equation}
A numerical evaluation of the distribution function of $\tau_{a}$ (and more
generally $\tau_{a}^{n}$ and $\tilde{\tau}_{a}^{n}$) by an inverse Laplace
transform method will be given at the end of Section 4. Other forms of
infinite series expansion of the distribution of $\tau_{a}$ were derived by
Douady et al. \cite{DSY00} and Magdon et al. \cite{MDD04} for a standard
Brownian motion and a drifted Brownian motion, respectively. By taking the
derivative with respect to $\lambda$ in (\ref{lap}) and letting $\lambda
\rightarrow0+$, we have
\[
\mathbb{E}[\tau_{a}]=\frac{\sigma^{2}\mathrm{e}^{2\mu a/\sigma^{2}}-\sigma
^{2}-2\mu a}{2\mu^{2}}.
\]
It is straightforward to check that
\begin{equation}
\lim_{\lambda\rightarrow0+}b_{\lambda}=\lim_{\lambda\rightarrow0+}c_{\lambda
}=\frac{\gamma}{\mathrm{e}^{\gamma a}-1},\label{bc}%
\end{equation}
where $\gamma=\frac{2\mu}{\sigma^{2}}$. In the risk theory literature, the
constant $\gamma$ is known as the \emph{adjustment coefficient}. In
particular, when $\mu=0$, the quantity $\frac{\gamma}{\mathrm{e}^{\gamma a}%
-1}$ is understood as $\lim_{\gamma\rightarrow0}\frac{\gamma}{\mathrm{e}%
^{\gamma a}-1}=\frac{1}{a}$. It follows from (\ref{lap}) and (\ref{bc}) that
\[
\mathbb{P}\left\{  \tau_{a}<\infty\right\}  =\lim_{\lambda\rightarrow
0+}\mathbb{E}\left[  \mathrm{e}^{-\lambda\tau_{a}}\right]  =1.
\]
Furthermore, we have
\begin{equation}
\mathbb{P}\left\{  M_{\tau_{a}}\geq x\right\}  =\mathbb{P}\left\{  M_{\tau
_{a}}\geq x,\tau_{a}<\infty\right\}  =\lim_{\lambda\rightarrow0+}%
\mathbb{E}\left[  \mathrm{e}^{-\lambda\tau_{a}};M_{\tau_{a}}\geq x\right]
=\mathrm{e}^{-\frac{\gamma x}{\mathrm{e}^{\gamma a}-1}}.\label{M}%
\end{equation}
which implies that the running maximum at the first drawdown time $M_{\tau
_{a}}$ follows an exponential distribution with mean $\left(  \mathrm{e}%
^{\gamma a}-1\right)  /\gamma$ (see, e.g., Lehoczky \cite{Lehoczky77}).

\section{The drawdown times with recovery}

We begin our analysis with the drawdown times with recovery $\{\tilde{\tau
}_{a}^{n},n\in%
\mathbb{N}
\}$ given that their structure leads to a simpler analysis than their
counterpart ones without recovery.

We first consider the asymptotic behavior of the frequency rate of drawdowns
with recovery. Let $\tilde{N}_{t}^{a}=\sum\nolimits_{n=1}^{\infty}1_{\left\{
\tilde{\tau}_{a}^{n}\leq t\right\}  }$ be the number of drawdowns with
recovery observed by time $t\geq0$, and define $\tilde{N}_{t}^{a}/t$ to be the
frequency rate of drawdowns. It is clear that $\left\{  \tilde{N}_{t}%
^{a},t\geq0\right\}  $ is a delayed renewal process where the first drawdown
time is distributed as $\tau_{a}$, while the subsequent inter-drawdown times
are independent and identically distributed as $T_{X_{\tau_{a}}+a}^{+}%
\circ\tau_{a}$. From Theorem 6.1.1 of Rolski et al. \cite{Rolskibook}, it
follows that, with probability one,
\[
\lim_{t\rightarrow\infty}\frac{\tilde{N}_{t}^{a}}{t}=\left\{
\begin{array}
[c]{lc}%
\frac{1}{\mathbb{E}[\tau_{a}]+\mathbb{E}[T_{a}^{+}]}=\frac{2\mu^{2}}%
{\sigma^{2}\left(  \mathrm{e}^{2\mu a/\sigma^{2}}-1\right)  }, & \text{if }%
\mu>0,\\
0, & \text{if }\mu\leq0.
\end{array}
\right.
\]
Moreover, one could easily obtain some central limit theorems for $\tilde
{N}_{t}^{a}$ by Theorem 6.1.2 of Rolski et al. \cite{Rolskibook}.

Next, we study the joint Laplace transform of $\tilde{\tau}_{a}^{n}$ and
$M_{\tilde{\tau}_{a}^{n}}$. Note that $X_{\tilde{\tau}_{a}^{n}}=M_{\tilde
{\tau}_{a}^{n}}-a$ a.s. whenever $\tilde{\tau}_{a}^{n}<\infty$, and thus the
following theorem is sufficient to characterize the triplet $\left(
\tilde{\tau}_{a}^{n},M_{\tilde{\tau}_{a}^{n}},X_{\tilde{\tau}_{a}^{n}}\right)
$.

\begin{theorem}
\label{thm M til L}For $n\in%
\mathbb{N}
$ and $\lambda,x\geq0$, we have%
\begin{equation}
\mathbb{E}\left[  \mathrm{e}^{-\lambda\tilde{\tau}_{a}^{n}};M_{\tilde{\tau
}_{a}^{n}}>x\right]  =\left(  \frac{c_{\lambda}}{b_{\lambda}}\right)
^{n}\mathrm{e}^{-(n-1)\beta_{\lambda}^{+}a}\sum_{m=0}^{n-1}\frac{(b_{\lambda
}x)^{m}}{m!}\mathrm{e}^{-b_{\lambda}x}. \label{M tau til L}%
\end{equation}

\end{theorem}

\begin{proof}
To prove this result, we first condition on the first drawdown time $\tau_{a}$
and subsequently on the time for the process $X$ to recover its running
maximum. Using the strong Markov property of $X$ and (\ref{JL}), it is clear
that
\begin{align}
\mathbb{E}\left[  \mathrm{e}^{-\lambda\tilde{\tau}_{a}^{n}-sM_{\tilde{\tau
}_{a}^{n}}}\right]   &  =\mathbb{E}\left[  \mathrm{e}^{-\lambda\tilde{\tau
}_{a}^{n}-sM_{\tilde{\tau}_{a}^{n}}};\tilde{\tau}_{a}^{n}<\infty\right]
\nonumber\\
&  =\mathbb{E}\left[  \mathrm{e}^{-\lambda\tau_{a}-sM_{\tau_{a}}}\right]
\mathbb{E}\left[  \mathrm{e}^{-T_{a}^{+}}\right]  \mathbb{E}\left[
\mathrm{e}^{-\lambda\tilde{\tau}_{a}^{n-1}-sM\tilde{\tau}_{a}^{n-1}}\right]
\nonumber\\
&  =\frac{c_{\lambda}}{b_{\lambda}+s}\mathrm{e}^{-\beta_{\lambda}^{+}%
a}\mathbb{E}\left[  \mathrm{e}^{-\lambda\tilde{\tau}_{a}^{n-1}-sM\tilde{\tau
}_{a}^{n-1}}\right] \nonumber\\
&  =\left(  \frac{c_{\lambda}}{b_{\lambda}+s}\right)  ^{n-1}\mathrm{e}%
^{-(n-1)\beta_{\lambda}^{+}a}\mathbb{E}\left[  \mathrm{e}^{-\lambda\tau
_{a}-sM_{\tau_{a}}}\right] \nonumber\\
&  =\left(  \frac{c_{\lambda}}{b_{\lambda}+s}\right)  ^{n}\mathrm{e}%
^{-(n-1)\beta_{\lambda}^{+}a}. \label{babu}%
\end{align}
Given that $\left(  b_{\lambda}/\left(  b_{\lambda}+s\right)  \right)  ^{n}$
is the Laplace transform of an Erlang random variable (rv) with mean
$n/b_{\lambda}$ and variance $n/\left(  b_{\lambda}\right)  ^{2}$, a tail
inversion of (\ref{babu}) wrt $s$ yields (\ref{M tau til L}). \bigskip
\end{proof}

In particular, letting $x\rightarrow0+$, we have%
\begin{equation}
\mathbb{E}\left[  \mathrm{e}^{-\lambda\tilde{\tau}_{a}^{n}}\right]  =\left(
c_{\lambda}/b_{\lambda}\right)  ^{n}\mathrm{e}^{-(n-1)\beta_{\lambda}^{+}a},
\label{L tau til}%
\end{equation}
for $n\in%
\mathbb{N}
$. Furthermore, letting $\lambda\rightarrow0+$ in (\ref{L tau til}), together
with (\ref{bc}) and $\lim_{\lambda\rightarrow0+}\beta_{\lambda}^{+}=\frac
{-\mu+|\mu|}{\sigma^{2}}$, we have
\begin{equation}
\mathbb{P}\left\{  \tilde{\tau}_{a}^{n}<\infty\right\}  =\left\{
\begin{array}
[c]{lc}%
1, & \text{if }\mu\geq0,\\
\mathrm{e}^{(n-1)\gamma a}, & \text{if }\mu<0.
\end{array}
\right.  \label{tau til P}%
\end{equation}
In other words, a historical running maximum may never be recovered if the
drift $\mu<0$.\bigskip

\begin{corollary}
For $n\in%
\mathbb{N}
$ and $x>0$, we have%
\begin{equation}
\mathbb{P}\left\{  M_{\tilde{\tau}_{a}^{n}}>x,\tilde{\tau}_{a}^{n}%
<\infty\right\}  =\left\{
\begin{array}
[c]{lc}%
\mathrm{e}^{-\frac{\gamma x}{\mathrm{e}^{\gamma a}-1}}\sum_{m=0}^{n-1}\frac
{1}{m!}\left(  \frac{\gamma x}{\mathrm{e}^{\gamma a}-1}\right)  ^{m}, &
\text{if }\mu\geq0,\\
\mathrm{e}^{(n-1)\gamma a}\mathrm{e}^{-\frac{\gamma x}{\mathrm{e}^{\gamma
a}-1}}\sum_{m=0}^{n-1}\frac{1}{m!}\left(  \frac{\gamma x}{\mathrm{e}^{\gamma
a}-1}\right)  ^{m}, & \text{if }\mu<0.
\end{array}
\right.  \text{.} \label{M tau til P}%
\end{equation}

\end{corollary}

\begin{proof}
Substituting (\ref{L tau til}) into (\ref{M tau til L})\ yields
\begin{equation}
\mathbb{E}\left[  \mathrm{e}^{-\lambda\tilde{\tau}_{a}^{n}};M_{\tilde{\tau
}_{a}^{n}}>x\right]  =\mathbb{E}\left[  \mathrm{e}^{-\lambda\tilde{\tau}%
_{a}^{n}}\right]  \sum_{m=0}^{n-1}\frac{(b_{\lambda}x)^{m}}{m!}\mathrm{e}%
^{-b_{\lambda}x}. \label{pep}%
\end{equation}
Taking the limit when $\lambda\rightarrow0+$ in (\ref{pep}), and then using
(\ref{bc}), one arrives at%
\begin{equation}
\mathbb{P}\left\{  M_{\tilde{\tau}_{a}^{n}}>x,\tilde{\tau}_{a}^{n}%
<\infty\right\}  =\mathbb{P}\left\{  \tilde{\tau}_{a}^{n}<\infty\right\}
\sum_{m=0}^{n-1}\frac{(\frac{\gamma x}{\mathrm{e}^{\gamma a}-1})^{m}}%
{m!}\mathrm{e}^{-\frac{\gamma x}{\mathrm{e}^{\gamma a}-1}}\text{.}
\label{pep1}%
\end{equation}
Substituting (\ref{tau til P}) into (\ref{pep1}) results in (\ref{M tau til P}%
).\bigskip
\end{proof}

Note that (\ref{pep1}) indicates
\begin{equation}
\mathbb{P}\left\{  M_{\tilde{\tau}_{a}^{n}}>x\left\vert \tilde{\tau}_{a}%
^{n}<\infty\right.  \right\}  =\sum_{m=0}^{n-1}\frac{1}{m!}\left(
\frac{\gamma x}{\mathrm{e}^{\gamma a}-1}\right)  ^{m}\mathrm{e}^{-\frac{\gamma
x}{\mathrm{e}^{\gamma a}-1}}\text{,} \label{abc}%
\end{equation}
for all $\mu\in\mathbb{R}$. This result can be interpreted probabilistically.
Indeed, when $\tilde{\tau}_{a}^{n}<\infty$, $M_{\tilde{\tau}_{a}^{m}%
}-M_{\tilde{\tau}_{a}^{m-1}}$ follows an exponential distribution with mean
$\left(  \mathrm{e}^{\gamma a}-1\right)  /\gamma$ for $m=1,2,...,n$. From the
strong Markov property, the rv's $M_{\tilde{\tau}_{a}^{m}}-M_{\tilde{\tau}%
_{a}^{m-1}}$ for all $m=1,2,...,n$ are all independent, and thus
$M_{\tilde{\tau}_{a}^{n}}=\sum_{m=1}^{n}\left(  M_{\tilde{\tau}_{a}^{m}%
}-M_{\tilde{\tau}_{a}^{m-1}}\right)  $ is an Erlang rv with survival function
(\ref{abc}).

In particular, when $n\rightarrow\infty$, it is easy to check that
$\lim_{n\rightarrow\infty}\mathbb{P}\left\{  M_{\tilde{\tau}_{a}^{n}%
}>x\right\}  =\mathbb{P}\left\{  T_{x}^{+}<\infty\right\}  $ which agrees with
(\ref{one-sided P}). For completeness, we conclude this section with a result
that is immediate from (\ref{M tau til L}) and the fact that $M_{\tilde{\tau
}_{a}^{n}}-X_{\tilde{\tau}_{a}^{n}}=a$ a.s. whenever $\tilde{\tau}_{a}%
^{n}<\infty$.

\begin{corollary}
For $n\in%
\mathbb{N}
$ and $x\geq-a$, we have%
\[
\mathbb{E}\left[  \mathrm{e}^{-\lambda\tilde{\tau}_{a}^{n}};X_{\tilde{\tau
}_{a}^{n}}>x\right]  =\left(  \frac{c_{\lambda}}{b_{\lambda}}\right)
^{n}\mathrm{e}^{-(n-1)\beta_{\lambda}^{+}a}\sum_{m=0}^{n-1}\frac{\left(
b_{\lambda}(x+a)\right)  ^{m}}{m!}\mathrm{e}^{-b_{\lambda}(x+a)}.
\]

\end{corollary}

\section{Drawdown times without recovery}

In this section, we focus on the drawdown times without recovery which are
more challenging to analyze than their counterparts with recovery.

Let $N_{t}^{a}=\sum\nolimits_{n=1}^{\infty}1_{\left\{  \tau_{a}^{n}\leq
t\right\}  }$ be the number of drawdowns without recovery by time $t\geq0$.
Clearly, $\left\{  N_{t}^{a},t\geq0\right\}  $ is a renewal process with
independent inter-drawdown times, all distributed as $\tau_{a}$. By Theorem
6.1.1 of Rolski et al. \cite{Rolskibook}, it follows that, with probability
one,%
\[
\lim_{t\to\infty}\frac{N_{t}^{a}}{t}=\frac{1}{\mathbb{E}\left[  \tau
_{a}\right]  }=\frac{2\mu^{2}}{\sigma^{2}\mathrm{e}^{2\mu a/\sigma^{2}}%
-\sigma^{2}-2\mu a}\text{,}%
\]
which is consistent with our intuition based on (\ref{iterated}). Here again,
one can also obtain some central limit theorems for $N_{t}^{a}$ by an
application of Theorem 6.1.2 of Rolski et al. \cite{Rolskibook}.

Next, we characterize the joint distribution of $\left(  \tau_{a}^{n}%
,X_{\tau_{a}^{n}}\right)  $ by deriving an explicit expression for
$\mathbb{E}[\mathrm{e}^{-\lambda\tau_{a}^{n}};X_{\tau_{a}^{n}}>x]$.

\begin{theorem}
\label{yyz}For $n\in%
\mathbb{N}
$ and $\lambda,x>0$, the joint distribution of $\left(  \tau_{a}^{n}%
,X_{\tau_{a}^{n}}\right)  $ satisfies
\begin{equation}
\mathbb{E}[\mathrm{e}^{-\lambda\tau_{a}^{n}};X_{\tau_{a}^{n}}>x]=\left(
\frac{c_{\lambda}}{b_{\lambda}}\right)  ^{n}\mathrm{e}^{-b_{\lambda}%
(x+na)}\sum_{m=0}^{n-1}\frac{\left(  b_{\lambda}(x+na)\right)  ^{m}}{m!}.
\label{L tau x}%
\end{equation}

\end{theorem}

\begin{proof}
Given that $X_{\tau_{a}^{n}}+na$ is a positive rv (and $X_{\tau_{a}^{n}}$ is
not), we prove (\ref{L tau x}) by first deriving an expression for the joint
Laplace transform of $\left(  \tau_{a}^{n},X_{\tau_{a}^{n}}+na\right)  $. By
conditioning on the first drawdown time and its associated value process, and
by making use of the strong Markov property and (\ref{JL}), it is clear that
for all $s\ge0$,
\begin{align}
\mathbb{E}\left[  \mathrm{e}^{-\lambda\tau_{a}^{n}-s\left(  X_{\tau_{a}^{n}%
}+na\right)  }\right]   &  =\mathbb{E}\left[  \mathrm{e}^{-\lambda\tau
_{a}-s\left(  X_{\tau_{a}}+a\right)  }\right]  \mathbb{E}\left[
\mathrm{e}^{-\lambda\tau_{a}^{n-1}-s\left(  X_{\tau_{a}^{n-1}}+\left(
n-1\right)  a\right)  }\right] \nonumber\\
&  =\mathbb{E}\left[  \mathrm{e}^{-\lambda\tau_{a}-sM_{\tau_{a}}}\right]
\mathbb{E}\left[  \mathrm{e}^{-\lambda\tau_{a}^{n-1}-s\left(  X_{\tau
_{a}^{n-1}}+\left(  n-1\right)  a\right)  }\right] \nonumber\\
&  =\frac{c_{\lambda}}{b_{\lambda}+s}\mathbb{E}\left[  \mathrm{e}%
^{-\lambda\tau_{a}^{n-1}-s\left(  X_{\tau_{a}^{n-1}}+\left(  n-1\right)
a\right)  }\right] \nonumber\\
&  =\left(  \frac{c_{\lambda}}{b_{\lambda}+s}\right)  ^{n}\text{.} \label{aaa}%
\end{align}
The Laplace transform inversion of (\ref{aaa}) with respect to $s$ results in%
\begin{equation}
\mathbb{E}\left[  \mathrm{e}^{-\lambda\tau_{a}^{n}};\left(  X_{\tau_{a}^{n}%
}+na\right)  \in\mathrm{d} y\right]  =\left(  c_{\lambda}\right)  ^{n}%
\frac{y^{n-1}\mathrm{e}^{-b_{\lambda}y}}{\left(  n-1\right)  !}\mathrm{d}
y\text{,} \label{aaa1}%
\end{equation}
for $y\geq0$. Integrating (\ref{aaa1}) over $y$ from $x+na$ to $\infty$ yields
(\ref{L tau x}).\bigskip
\end{proof}

Letting $s\rightarrow0+$ in (\ref{aaa}), it follows that
\begin{equation}
\mathbb{E}[\mathrm{e}^{-\lambda\tau_{a}^{n}}]=\left(  c_{\lambda}/b_{\lambda
}\right)  ^{n}=\left(  \mathbb{E}[\mathrm{e}^{-\lambda\tau_{a}}]\right)  ^{n}.
\label{L tau}%
\end{equation}
Note that (\ref{L tau}) and (\ref{bc}) implies that%
\[
\mathbb{P}\left\{  \tau_{a}^{n}<\infty\right\}  =1.
\]
It is worth pointing out that the relation $\mathbb{E}\left[  \mathrm{e}%
^{-\lambda\tau_{a}^{n}}\right]  =\left(  \mathbb{E}\left[  \mathrm{e}%
^{-\lambda\tau_{a}}\right]  \right)  ^{n}$ holds more generally for $X$ a
general L\'{e}vy process or a renewal risk process (also known as the Sparre
Andersen risk model \cite{AndersenRiskmodel}) given that the inter-drawdown
times $\tau_{a}^{1}$, and $\left\{  \tau_{a}^{n}-\tau_{a}^{n-1}\right\}
_{n\geq2}$\ form a sequence of i.i.d. rvs.

Similarly, letting $\lambda\rightarrow0+$ in (\ref{L tau x}), it follows that
\begin{equation}
\mathbb{P}\left\{  X_{\tau_{a}^{n}}\geq x\right\}  =\mathrm{e}^{-\frac
{\gamma(x+na)}{\mathrm{e}^{\gamma a}-1}}\sum_{m=0}^{n-1}\frac{\left(
\frac{\gamma(x+na)}{\mathrm{e}^{\gamma a}-1}\right)  ^{m}}{m!},
\label{X tau P}%
\end{equation}
for $n\in%
\mathbb{N}
$ and $x\geq-na$. As expected, (\ref{X tau P}) is the survival function of an
Erlang rv with mean $n\left(  \mathrm{e}^{\gamma a}-1\right)  /\gamma$ and
variance $n\left(  \left(  \mathrm{e}^{\gamma a}-1\right)  /\gamma\right)
^{2}$, later translated by $-na$ units.

Our objective is now to include $M_{\tau_{a}^{n}}$ in the analysis of the
$n$-th drawdown time. A result particularly useful to do so is provided in
Lemma \ref{constLT} which consider a specific constrained Laplace transform of
the first passage time to level $x$.

\begin{lemma}
\label{constLT}For $n\in%
\mathbb{N}
$ and $x>0$, the constrained Laplace transform of $T_{x}^{+}$ together with
this first passage time occurring before $\tau_{a}^{n}$ is given by%
\begin{equation}
\mathbb{E}\left[  \mathrm{e}^{-\lambda T_{x}^{+}};T_{x}^{+}<\tau_{a}%
^{n}\right]  =\mathrm{e}^{-b_{\lambda}x}\sum_{j=0}^{n-1}\left(  c_{\lambda
}\mathrm{e}^{-b_{\lambda}a}\right)  ^{j}\frac{x(x+ja)^{j-1}}{j!}\text{.}
\label{cLT}%
\end{equation}

\end{lemma}

\begin{proof}
We prove this result by induction on $n$. For $n=1$, we have%
\begin{align*}
\mathbb{E}\left[  \mathrm{e}^{-\lambda T_{x}^{+}};T_{x}^{+}<\tau_{a}%
^{1}\right]   &  =\mathbb{E}\left[  e^{-\lambda T_{x}^{+}}\right]
-\mathbb{E}\left[  e^{-\lambda T_{x}^{+}};T_{x}^{+}>\tau_{a}^{1}\right] \\
&  =\mathrm{e}^{-\beta_{\lambda}^{+}x}-\int_{0}^{x}\mathbb{E}\left[
\mathrm{e}^{-\lambda\tau_{a}^{1}}; M_{\tau_{a}^{1}}\in\mathrm{d} y\right]
\,\mathbb{E}_{y-a}\left[  \mathrm{e}^{-\lambda T_{x}^{+}}\right] \\
&  =\mathrm{e}^{-\beta_{\lambda}^{+}x}-\int_{0}^{x}c_{\lambda}\mathrm{e}%
^{-b_{\lambda}y}\,\mathrm{e}^{-\beta_{\lambda}^{+}\left(  x-y+a\right)
}\mathrm{d} y\\
&  =\mathrm{e}^{-\beta_{\lambda}^{+}x}-c_{\lambda}\mathrm{e}^{-\beta_{\lambda
}^{+}a}\frac{\mathrm{e}^{-\beta_{\lambda}^{+}x}-\mathrm{e}^{-b_{\lambda}x}%
}{b_{\lambda}-\beta_{\lambda}^{+}}\text{,}%
\end{align*}
where we used \eqref{L1} in the third equality.

On the other hand, using the fact that $c_{\lambda}\mathrm{e}^{-\beta
_{\lambda}^{+}a}=b_{\lambda}-\beta_{\lambda}^{+}$, we have
\[
\mathbb{E}\left[  \mathrm{e}^{-\lambda T_{x}^{+}};T_{x}^{+}<\tau_{a}%
^{1}\right]  =\mathrm{e}^{-b_{\lambda}x}\text{.}%
\]

We now assume that (\ref{cLT}) holds for $n=1,2,...,k-1$ and shows that
(\ref{cLT})\ also holds for $n=k$. Indeed, by the total probability formula,%
\begin{align}
\mathbb{E}\left[  \mathrm{e}^{-\lambda T_{x}^{+}};T_{x}^{+}<\tau_{a}%
^{k}\right]   &  =\mathbb{E}\left[  \mathrm{e}^{-\lambda T_{x}^{+}};T_{x}%
^{+}<\tau_{a}^{1}\right]  +\mathbb{E}\left[  \mathrm{e}^{-\lambda T_{x}^{+}%
};\tau_{a}^{1}<T_{x}^{+}<\tau_{a}^{k}\right] \nonumber\\
&  =\mathrm{e}^{-b_{\lambda}x}+\int_{0}^{x}\mathbb{E}\left[  \mathrm{e}%
^{-\lambda\tau_{a}};M_{\tau_{a}}\in\mathrm{d}y\right]  \mathbb{E}_{y-a}\left[
\mathrm{e}^{-\lambda T_{x}^{+}};T_{x}^{+}<\tau_{a}^{k-1}\right]
\mathrm{d}y\nonumber\\
&  =\mathrm{e}^{-b_{\lambda}x}+\int_{0}^{x}c_{\lambda}\mathrm{e}^{-b_{\lambda
}y}\,\mathbb{E}\left[  \mathrm{e}^{-\lambda T_{x-y+a}^{+}};T_{x-y+a}^{+}%
<\tau_{a}^{k-1}\right]  \mathrm{d}y\text{.} \label{parta}%
\end{align}
Substituting (\ref{cLT}) at $n=k-1$\ into (\ref{parta}) yields%
\begin{align*}
&  \mathbb{E}\left[  \mathrm{e}^{-\lambda T_{x}^{+}};T_{x}^{+}<\tau_{a}%
^{k}\right] \\
&  =\mathrm{e}^{-b_{\lambda}x}+c_{\lambda}\mathrm{e}^{-b_{\lambda}\left(
x+a\right)  }\sum_{j=0}^{k-2}\int_{0}^{x}\left(  c_{\lambda}\mathrm{e}%
^{-b_{\lambda}a}\right)  ^{j}\frac{\left(  x-y+a\right)  (x-y+\left(
j+1\right)  a)^{j-1}}{j!}\mathrm{d}y\\
&  =\mathrm{e}^{-b_{\lambda}x}+c_{\lambda}\mathrm{e}^{-b_{\lambda}\left(
x+a\right)  }\left(  x+\sum_{j=1}^{k-2}\left(  c_{\lambda}\mathrm{e}%
^{-b_{\lambda}a}\right)  ^{j}\int_{0}^{x}\left(  \frac{\left(  y+\left(
j+1\right)  a\right)  ^{j}}{j!}-a\frac{\left(  y+\left(  j+1\right)  a\right)
^{j-1}}{\left(  j-1\right)  !}\right)  \mathrm{d}y\right) \\
&  =\mathrm{e}^{-b_{\lambda}x}\left(  1+c_{\lambda}\mathrm{e}^{-b_{\lambda}%
a}x+\sum_{j=2}^{k-1}\left(  c_{\lambda}\mathrm{e}^{-b_{\lambda}a}\right)
^{j}\frac{x\left(  x+ja\right)  ^{j-1}}{j!}\right) \\
&  =\mathrm{e}^{-b_{\lambda}x}\sum_{j=0}^{k-1}\left(  c_{\lambda}%
\mathrm{e}^{-b_{\lambda}a}\right)  ^{j}\frac{x(x+ja)^{j-1}}{j!}\text{.}%
\end{align*}
This completes the proof.\bigskip
\end{proof}

In the next theorem, we provide a distributional characterization of the
$n$-th drawdown time $\tau_{a}^{n}$ with respect to both $M_{\tau_{a}^{n}}$
and $X_{\tau_{a}^{n}}$.

\begin{theorem}
\label{jointd}For $n\in%
\mathbb{N}
$ and $x>0$, we have
\begin{align}
&  \mathbb{E}\left[  \mathrm{e}^{-\lambda\tau_{a}^{n}};M_{\tau_{a}^{n}%
}>x,X_{\tau_{a}^{n}}\in\mathrm{d}y\right] \nonumber\\
&  =\left(  c_{\lambda}\right)  ^{n}\mathrm{e}^{-b_{\lambda}(y+na)}\sum
_{m=0}^{n-1}\frac{x(x+ma)^{m-1}(y-x+(n-m)a))^{n-1-m}\mathrm{1}_{\left\{
y-x+(n-m)a\geq0\right\}  }}{m!(n-m-1)!}\mathrm{d}y\text{.} \label{JJ}%
\end{align}

\end{theorem}

\begin{proof}
By conditioning on the drawdown episode during which the drifted Brownian
motion process $X$ reaches level $x$ for the first time and subsequently using
the strong Markov property, we have
\begin{align}
&  \mathbb{E}\left[  \mathrm{e}^{-\lambda\tau_{a}^{n}};M_{\tau_{a}^{n}%
}>x,X_{\tau_{a}^{n}}\in\mathrm{d}y\right] \nonumber\\
&  =\sum_{m=0}^{n-1}\mathbb{E}\left[  \mathrm{e}^{-\lambda\tau_{a}^{n}%
};M_{\tau_{a}^{n}}>x,X_{\tau_{a}^{n}}\in\mathrm{d}y,\tau_{a}^{m}<T_{x}%
^{+}<\tau_{a}^{m+1}\right] \nonumber\\
&  =\sum_{m=0}^{n-1}\mathbb{E}\left[  \mathrm{e}^{-\lambda T_{x}^{+}};\tau
_{a}^{m}<T_{x}^{+}<\tau_{a}^{m+1}\right]  \mathbb{E}_{x}\left[  \mathrm{e}%
^{-\lambda\tau_{a}^{n-m}};X_{\tau_{a}^{n-m}}\in\mathrm{d}y\right]  \label{J1}%
\end{align}
From Lemma \ref{constLT}, we know that%
\begin{align}
\mathbb{E}\left[  \mathrm{e}^{-\lambda T_{x}^{+}};\tau_{a}^{m}<T_{x}^{+}%
<\tau_{a}^{m+1}\right]   &  =\mathbb{E}\left[  \mathrm{e}^{-\lambda T_{x}^{+}%
};\tau_{a}^{m}<T_{x}^{+}\right]  -\mathbb{E}\left[  \mathrm{e}^{-\lambda
T_{x}^{+}};\tau_{a}^{m+1}<T_{x}^{+}\right] \nonumber\\
&  =\left(  c_{\lambda}\right)  ^{m}\frac{x(x+ma)^{m-1}}{m!}\mathrm{e}%
^{-b_{\lambda}\left(  x+ma\right)  }\text{.} \label{J2}%
\end{align}
By Theorem \ref{yyz}, we have
\begin{align}
&  \mathbb{E}_{x}\left[  \mathrm{e}^{-\lambda\tau_{a}^{n-m}};X_{\tau_{a}%
^{n-m}}\in\mathrm{d}y\right] \nonumber\\
&  =\frac{\left(  c_{\lambda}\right)  ^{n-m}(y-x+(n-m)a)^{n-m-1}%
\mathrm{e}^{-b_{\lambda}(y-x+(n-m)a)}1_{\left\{  y-x+(n-m)a\geq0\right\}  }%
}{(n-m-1)!}\mathrm{d}y. \label{J3}%
\end{align}
Substituting (\ref{J2}) and (\ref{J3}) into (\ref{J1}) and simplifying, one
easily obtains (\ref{JJ}).\bigskip
\end{proof}

Recall that $\tau_{a}^{1}=\tilde{\tau}_{a}^{1}=\tau_{a}$ and $X_{\tau_{a}%
}=M_{\tau_{a}}-a$ a.s.. Therefore, by letting $\lambda\rightarrow0+$ and $x=a$
in (\ref{J2}), it follows that, for $m=0,1,2,\cdots$,%
\begin{align}
\mathbb{P}\left\{  \tilde{\tau}_{a}^{2}=\tau_{a}^{2+m}\right\}   &
=\mathbb{P}\{\tau_{a}^{m}<T_{a}^{+}<\tau_{a}^{m+1}\}\nonumber\\
&  =\frac{(m+1)^{m-1}}{m!}\left(  \frac{\gamma a}{\mathrm{e}^{\gamma a}%
-1}\right)  ^{m}\mathrm{e}^{-\frac{\left(  m+1\right)  \gamma a}%
{\mathrm{e}^{\gamma a}-1}}, \label{eq:f2}%
\end{align}
which is the probability mass function of a generalized Poisson rv (see, e.g.,
Equation (9.1) of Consul and Famoye \cite{Lagrangian2006} with $\theta
=\lambda=\gamma a/(\mathrm{e}^{\gamma a}-1)$). For completeness, a rv $Y$ has
a generalized Poisson$\left(  \theta,\lambda\right)  $ distribution if its
probability mass function $p_{Y}$ is given by
\[
p_{Y}\left(  m\right)  =\frac{\theta\left(  \theta+\lambda m\right)
^{m-1}e^{-\theta-\lambda m}}{m!}\text{,\qquad}m=0,1,2,...\text{,}%
\]
when both $\theta,\lambda>0$.

Note that a generalization of (\ref{eq:f2}) will be proposed in Theorem
\ref{thm tt}.

\begin{remark}
\label{rk dd}Equation (\ref{eq:f2}) can be interpreted as follows: the number
of drawdowns \textbf{without} recovery between two successive drawdowns with
recovery follows a generalized Poisson distribution with $\theta
=\lambda=\gamma a/(\mathrm{e}^{\gamma a}-1)$.
\end{remark}

The following result connecting the two drawdown time sequences is provided.
It should be noted that the rv $N_{\tilde{\tau}_{a}^{k}}^{a}-k$ represents the
number of drawdowns without recovery over the first $k$ drawdowns with
recovery. When $k=2$, (\ref{allo}) coincides with (\ref{eq:f2}).

\begin{theorem}
\label{thm tt}For any $k\in%
\mathbb{N}
$, $N_{\tilde{\tau}_{a}^{k}}^{a}-k$ follows a generalized Poisson distribution
with parameters $\theta=(k-1)\gamma a/($\textrm{$e$}$^{\gamma a}-1)$ and
$\lambda=\gamma a/(\mathrm{e}^{\gamma a}-1)$, i.e., for $m=0,1,2,\ldots,$ we
have%
\begin{equation}
\mathbb{P}\left\{  \tilde{\tau}_{a}^{k}=\tau_{a}^{k+m}\right\}  =\mathbb{P}%
\left\{  N_{\tilde{\tau}_{a}^{k}}^{a}=k+m\right\}  =\frac{k-1}{m+k-1}%
\frac{\left(  \frac{\left(  m+k-1\right)  \gamma a}{\mathrm{e}^{\gamma a}%
-1}\right)  ^{m}}{m!}\mathrm{e}^{-\frac{\left(  m+k-1\right)  \gamma
a}{\mathrm{e}^{\gamma a}-1}}. \label{allo}%
\end{equation}

\end{theorem}

\begin{proof}
It is clear that $\left\{  \tilde{\tau}_{a}^{k}=\tau_{a}^{k+m}\right\}
\ $corresponds to the event that $m$ drawdowns without recovery will occur
over the first $k$ drawdowns with recovery, i.e.
\[
\left\{  \tilde{\tau}_{a}^{k}=\tau_{a}^{k+m}\right\}  =\left\{  N_{\tilde
{\tau}_{a}^{k}}^{a}=k+m\right\}  \text{.}%
\]
Next we prove $N_{\tilde{\tau}_{a}^{k}}^{a}-k$ follows a generalized Poisson
distribution. By Remark \ref{rk dd} and the strong Markov property of $X$, we
know that the numbers of drawdowns without recovery between any two successive
drawdowns with recovery are i.i.d. and follow a generalized Poisson
distribution with $\theta=\lambda=\gamma a/(\mathrm{e}^{\gamma a}-1)$. Thus,
\[
N_{\tilde{\tau}_{a}^{k}}^{a}-k=\sum_{i=2}^{k}\left(  N_{\tilde{\tau}_{a}^{i}%
}^{a}-N_{\tilde{\tau}_{a}^{i-1}}^{a}-1\right)  \text{,}%
\]
corresponds to a sum of i.i.d. rv's with a generalized Poisson distribution
$\theta=\lambda=\gamma a/(\mathrm{e}^{\gamma a}-1)$. Using Theorem 9.1 of
Consul and Famoye \cite{Lagrangian2006}, we have that $N_{\tilde{\tau}_{a}%
^{k}}^{a}-k$ follows a generalized Poisson distribution with parameters
$\theta=(k-1)\gamma a/($\textrm{$e$}$^{\gamma a}-1)$ and $\lambda=\gamma
a/(\mathrm{e}^{\gamma a}-1)$.\bigskip
\end{proof}

Next, we propose the following corollary which can be viewed as an extension
to Taylor \cite{Taylor75} and Lehoczky \cite{Lehoczky77} from the first
drawdown case to the $n$-th drawdown without recovery.

\begin{corollary}
\label{aab}For $n\in%
\mathbb{N}
$ and $x>0$, we have%
\[
\mathbb{E}\left[  \mathrm{e}^{-\lambda\tau_{a}^{n}};M_{\tau_{a}^{n}}>x\right]
=\left(  \frac{c_{\lambda}}{b_{\lambda}}\right)  ^{n}\sum_{m=0}^{n-1}%
\frac{x(x+ma)^{m-1}b_{\lambda}^{m}}{m!}\mathrm{e}^{-b_{\lambda}\left(
ma+x\right)  }.
\]

\end{corollary}

\begin{proof}
Taking the integral of (\ref{JJ}) with respect to $y$ in $(-na,\infty)$, we
have%
\begin{align*}
\mathbb{E}\left[  \mathrm{e}^{-\lambda\tau_{a}^{n}};M_{\tau_{a}^{n}}>x\right]
&  =(c_{\lambda})^{n}\sum_{m=0}^{n-1}\frac{x(x+ma)^{m-1}}{m!(n-m-1)!}%
\int_{x-(n-m)a}^{\infty}\mathrm{e}^{-b_{\lambda}(y+na)}(y-x+(n-m)a)^{n-m-1}%
\mathrm{d}y\\
&  =(c_{\lambda})^{n}\sum_{m=0}^{n-1}\frac{x(x+ma)^{m-1}}{m!(n-m-1)!}\int
_{0}^{\infty}\mathrm{e}^{-b_{\lambda}(z+x+ma)}z^{n-m-1}\mathrm{d}z\\
&  =(c_{\lambda})^{n}\sum_{m=0}^{n-1}\frac{x(x+ma)^{m-1}}{m!(n-m-1)!}%
\mathrm{e}^{-b_{\lambda}(x+ma)}\int_{0}^{\infty}\mathrm{e}^{-b_{\lambda}%
z}z^{n-m-1}\mathrm{d}z\\
&  =(c_{\lambda})^{n}\sum_{m=0}^{n-1}\frac{x(x+ma)^{m-1}}{m!b_{\lambda}^{n-m}%
}\mathrm{e}^{-b_{\lambda}(x+ma)}.
\end{align*}
which completes the proof.\bigskip
\end{proof}

The marginal distribution of $M_{\tau_{a}^{n}}$ can easily be obtained from
Corollary \ref{aab} by letting $\lambda\rightarrow0+$ and subsequently making
use of (\ref{bc}). Indeed,
\begin{equation}
\mathbb{P}\left\{  M_{\tau_{a}^{n}}>x\right\}  =\sum_{m=0}^{n-1}%
\frac{x(x+ma)^{m-1}\left(  \frac{\gamma}{\mathrm{e}^{\gamma a}-1}\right)
^{m}}{m!}\mathrm{e}^{-\frac{\gamma(ma+x)}{\mathrm{e}^{\gamma a}-1}}\text{.}
\label{PP}%
\end{equation}
Rearrangements of (\ref{PP}) yields%
\begin{equation}
\mathbb{P}\left\{  M_{\tau_{a}^{n}}>x\right\}  =\sum_{k=0}^{n-1}D_{k,n}%
\frac{\left(  \frac{\gamma x}{\mathrm{e}^{\gamma a}-1}\right)  ^{k}}%
{k!}\mathrm{e}^{-\frac{\gamma x}{\mathrm{e}^{\gamma a}-1}}\text{,}
\label{survivalM}%
\end{equation}
where $D_{0,n}=1$, and%
\begin{equation}
D_{k,n}=\sum_{m=k}^{n-1}\frac{k\left(  \frac{m\gamma a}{\mathrm{e}^{\gamma
a}-1}\right)  ^{m-k}}{m\left(  m-k\right)  !}\mathrm{e}^{-\frac{m\gamma
}{\mathrm{e}^{\gamma a}-1}a}=\sum_{m=0}^{n-1-k}\frac{k\left(  \frac{\left(
m+k\right)  \gamma a}{\mathrm{e}^{\gamma a}-1}\right)  ^{m}}{(m+k)m!}%
\mathrm{e}^{-\frac{\left(  m+k\right)  \gamma a}{\mathrm{e}^{\gamma a}-1}%
}\text{,} \label{D1}%
\end{equation}
for $k=1,2,...,n-1$. Note that by substituting $k$ by $k+1$ in (\ref{allo}),
it follows that (\ref{D1}) can be rewritten as%
\[
D_{k,n}=\sum_{m=0}^{n-1-k}\mathbb{P}\left\{  \tilde{\tau}_{a}^{k+1}=\tau
_{a}^{k+1+m}\right\}  \text{,}%
\]
which is equivalent to
\[
D_{k,n}=\mathbb{P}\left\{  \tilde{\tau}_{a}^{k+1}\leq\tau_{a}^{n}\right\}
=\mathbb{P}\left\{  \tilde{N}_{\tau_{a}^{n}}^{a}>k\right\}  \text{.}%
\]

Then,
\[
\mathbb{P}\left\{  M_{\tau_{a}^{n}}\in\mathrm{d}y\right\}  =\sum_{k=1}%
^{n}d_{k,n}\frac{\left(  \frac{\gamma a}{\mathrm{e}^{\gamma a}-1}\right)
^{k}y^{k-1}e^{-\frac{\gamma a}{\mathrm{e}^{\gamma a}-1}y}}{\left(  k-1\right)
!}\mathrm{d}y\text{,}%
\]
where $\left\{  d_{k,n}\right\}  _{k=1}^{n}$ are given by%
\begin{align*}
d_{k,n}  &  \equiv D_{k-1,n}-D_{k,n}\\
&  =\sum_{j=k}^{n}\frac{k-1}{j-1}\frac{\left(  \frac{\left(  j-1\right)
\gamma a}{\mathrm{e}^{\gamma a}-1}\right)  ^{j-k}}{\left(  j-k\right)
!}\mathrm{e}^{-\frac{\left(  j-1\right)  \gamma a}{\mathrm{e}^{\gamma a}-1}%
}\left(  1-\sum_{m=0}^{n-j-1}\frac{(m+1)^{m-1}}{m!}\left(  \frac{\gamma
a}{\mathrm{e}^{\gamma a}-1}\right)  ^{m}\mathrm{e}^{-\frac{\left(  m+1\right)
\gamma a}{\mathrm{e}^{\gamma a}-1}}\right)  \text{.}%
\end{align*}
In conclusion, $M_{\tau_{a}^{n}}$ follows a mixed-Erlang distribution which is
an important class of distribution in risk management (see, e.g., Willmot and
Lin \cite{WillLin} for an extensive review of mixed Erlang distributions).

\begin{remark}
\label{mE}Note that the distribution of $M_{\tau_{a}^{n}}$ does not come as a
surprise. Indeed, one can obtain the structural form of the distribution of
$M_{\tau_{a}^{n}}$ by conditioning on $\tilde{N}_{\tau_{a}^{n}}^{a}$, namely
the number of drawdowns with recovery over the first $n$ drawdowns (without
recovery). Using the strong Markov property of the process $X$ and Equation
(\ref{M}), it follows that $M_{\tau_{a}^{n}}\left\vert \tilde{N}_{\tau_{a}%
^{n}}^{a}=m\right.  $ is an Erlang rv with mean $m\frac{\mathrm{e}^{\gamma
a}-1}{\gamma}$ and variance $m\left(  \frac{\mathrm{e}^{\gamma a}-1}{\gamma
}\right)  ^{2}$ for $m=1,2,...,n$. Thus, in (\ref{survivalM}), $D_{k,n}$ can
be interpreted as the survival function of $\tilde{N}_{\tau_{a}^{n}}^{a}$,
i.e.
\[
D_{k,n}=\mathbb{P}\left\{  \tilde{N}_{\tau_{a}^{n}}^{a}>k\right\}
=\mathbb{P}\left\{  \tilde{\tau}_{a}^{k+1}\leq\tau_{a}^{n}\right\}  .
\]

\end{remark}

The next corollary investigates the actual drawdown $M_{t}-X_{t}$ at
$t=\tau_{a}^{n}$.

\begin{corollary}
For $a\leq x\leq na$, we have
\begin{align*}
&  \mathbb{E}\left[  \mathrm{e}^{-\lambda\tau_{a}^{n}};M_{\tau_{a}^{n}%
}-X_{\tau_{a}^{n}}\leq x\right] \\
&  =(c_{\lambda})^{n}\mathrm{e}^{-b_{\lambda}(na-x)}\sum_{m=0}^{n-1}\left(
\frac{(na-x)^{m}}{b_{\lambda}^{n-m}m!}-\frac{\mathrm{1}_{\left\{
x\leq(n-m)a\right\}  }((n-m)a-x)^{n-m-1}\int_{0}^{\infty}\mathrm{e}%
^{-b_{\lambda}y}y(y+ma)^{m-1}\mathrm{d}y}{m!(n-m-1)!}\right)  .
\end{align*}

\end{corollary}

\begin{proof}
We have%
\begin{align}
&  \mathbb{E}\left[  \mathrm{e}^{-\lambda\tau_{a}^{n}};M_{\tau_{a}^{n}%
}-X_{\tau_{a}^{n}}>x\right] \nonumber\\
&  =\int_{-x}^{\infty}\mathbb{E}\left[  \mathrm{e}^{-\lambda\tau_{a}^{n}%
};M_{\tau_{a}^{n}}-X_{\tau_{a}^{n}}>x,X_{\tau_{a}^{n}}\in\mathrm{d}y\right]
+\mathbb{E}\left[  \mathrm{e}^{-\lambda\tau_{a}^{n}};M_{\tau_{a}^{n}}%
-X_{\tau_{a}^{n}}>x,X_{\tau_{a}^{n}}\leq-x\right] \nonumber\\
&  =\int_{-x}^{\infty}\mathbb{E}\left[  \mathrm{e}^{-\lambda\tau_{a}^{n}%
};M_{\tau_{a}^{n}}>x+y,X_{\tau_{a}^{n}}\in\mathrm{d}y\right]  +\mathbb{E}%
\left[  \mathrm{e}^{-\lambda\tau_{a}^{n}};X_{\tau_{a}^{n}}\leq-x\right]
\nonumber\\
&  =\int_{-x}^{\infty}\mathbb{E}\left[  \mathrm{e}^{-\lambda\tau_{a}^{n}%
};M_{\tau_{a}^{n}}>x+y,X_{\tau_{a}^{n}}\in\mathrm{d}y\right]  +\left(
c_{\lambda}/b_{\lambda}\right)  ^{n}\left(  1-\mathrm{e}^{-b_{\lambda}%
(na-x)}\sum_{m=0}^{n-1}\frac{\left(  b_{\lambda}(na-x)\right)  ^{m}}%
{m!}\right)  , \label{MX}%
\end{align}
where the last step is due to (\ref{L tau x}). Moreover, by Theorem
\ref{jointd}, the first term of (\ref{MX})%
\begin{align*}
&  \int_{-x}^{\infty}\mathbb{E}\left[  \mathrm{e}^{-\lambda\tau_{a}^{n}%
};M_{\tau_{a}^{n}}>x+y,X_{\tau_{a}^{n}}\in\mathrm{d}y\right] \\
&  =(c_{\lambda})^{n}\sum_{m=0}^{n-1}\frac{((n-m)a-x)^{n-m-1}\mathrm{1}%
_{\left\{  -x+(n-m)a\geq0\right\}  }}{m!(n-m-1)!}\int_{-x}^{\infty}%
\mathrm{e}^{-b_{\lambda}(y+na)}(x+y)(x+y+ma)^{m-1}\mathrm{d}y\\
&  =(c_{\lambda})^{n}\sum_{m=0}^{n-1}\frac{((n-m)a-x)^{n-m-1}\mathrm{1}%
_{\left\{  x\leq(n-m)a\right\}  }}{m!(n-m-1)!}\int_{0}^{\infty}\mathrm{e}%
^{-b_{\lambda}(z-x+na)}z(z+ma)^{m-1}\mathrm{d}z\\
&  =(c_{\lambda})^{n}\mathrm{e}^{-b_{\lambda}\left(  na-x\right)  }\sum
_{m=0}^{n-1}\frac{((n-m)a-x)^{n-m-1}\mathrm{1}_{\left\{  x\leq(n-m)a\right\}
}}{m!(n-m-1)!}\int_{0}^{\infty}\mathrm{e}^{-b_{\lambda}z}z(z+ma)^{m-1}%
\mathrm{d}z.
\end{align*}
Substituting this back into (\ref{MX}), we complete the proof.\bigskip
\end{proof}

To complete the section, we consider a numerical example to compare the
distribution of the $n$-th drawdown times $\tilde{\tau}_{a}^{n}$ and $\tau
_{a}^{n}$ whose Laplace transforms are given in (\ref{L tau til}) and
(\ref{L tau}), respectively. We implement a numerical inverse Laplace
transform approach proposed by Abate and Whitt \cite{AbatWhit06}. For ease of
notation, we denote the cumulative distribution functions of $\tau_{a}^{n}$
and $\tilde{\tau}_{a}^{n}$ by $F_{n}$ and $\tilde{F}_{n}$, respectively.

\begin{center}
\textbf{Table 4.1 }Distribution of the $n$-th drawdown times when $a=0.1$ and
$\sigma=0.2$%
\[%
\begin{tabular}
[c]{c|c|c|c}\hline
& $\mu=0.1$ & $\mu=0$ & $\mu=-0.1$\\\hline
\multicolumn{1}{l|}{%
\begin{tabular}
[c]{l}%
\\
$n=1$\\
$n=2$\\
$n=3$\\
$n=4$\\
$n=5$\\
$n=6$%
\end{tabular}
} & \multicolumn{1}{|l|}{%
\begin{tabular}
[c]{ll}%
$F_{n}(1)$ & $\tilde{F}_{n}(1)$\\
0.9779 & 0.9779\\
0.8759 & 0.4865\\
0.6651 & 0.1024\\
0.4060 & 0.0082\\
0.1942 & 0.0002\\
0.0721 & 0.0000
\end{tabular}
} & \multicolumn{1}{|l|}{%
\begin{tabular}
[c]{ll}%
$F_{n}(1)$ & $\tilde{F}_{n}(1)$\\
0.9908 & 0.9908\\
0.9366 & 0.4406\\
0.7926 & 0.0885\\
0.5652 & 0.0070\\
0.3262 & 0.0002\\
0.1492 & 0.0000
\end{tabular}
} & \multicolumn{1}{|l}{%
\begin{tabular}
[c]{ll}%
$F_{n}(1)$ & $\tilde{F}_{n}(1)$\\
0.9967 & 0.9967\\
0.9719 & 0.3636\\
0.8874 & 0.0663\\
0.7166 & 0.0050\\
0.4871 & 0.0001\\
0.2696 & 0.0000
\end{tabular}
}\\\hline
\end{tabular}
\
\]

\end{center}

Table 4.1 presents the probabilities that at least $n$ drawdowns with or
without recovery occurs before time $1$ for different values of the drift
$\mu$. We observe that $F_{n}(1)>\tilde{F}_{n}(1)$ for $n\geq2$ due to the
relation between $\tau_{a}^{n}$ and $\tilde{\tau}_{a}^{n}$ given in
(\ref{allo}). In addition, it shows that $F_{n}(1)$ increases as $\mu$
decreases. However, we observe the opposite trend for $\tilde{F}_{n}(1)$ when
$n\geq2$. This is because the previous running maximum is less likely to be
revisited for a smaller $\mu$. Since the drawdown risk is in principle a type
of downside risk, we think smaller $\mu$ should lead to higher downside risks.
In this sense, we suggest that the drawdown times without recovery are better
to capture the essence of drawdown risks.

\begin{center}
\textbf{Table 4.2 }Distribution of drawdown times when $a=0.1$ and
$\sigma=0.12$%
\[%
\begin{tabular}
[c]{c|c|c|c}\hline
& $\mu=0.1$ & $\mu=0$ & $\mu=-0.1$\\\hline
\multicolumn{1}{l|}{%
\begin{tabular}
[c]{l}%
\\
$n=1$\\
$n=2$\\
$n=3$\\
$n=4$\\
$n=5$\\
$n=6$%
\end{tabular}
} & \multicolumn{1}{|l|}{%
\begin{tabular}
[c]{ll}%
$F_{n}(1)$ & $\tilde{F}_{n}(1)$\\
0.5663 & 0.5663\\
0.1592 & 0.0339\\
0.0225 & 0.0002\\
0.0016 & 0.0000\\
0.0001 & 0.0000\\
0.0000 & 0.0000
\end{tabular}
} & \multicolumn{1}{|l|}{%
\begin{tabular}
[c]{ll}%
$F_{n}(1)$ & $\tilde{F}_{n}(1)$\\
0.7845 & 0.7845\\
0.3755 & 0.0494\\
0.0986 & 0.0002\\
0.0137 & 0.0000\\
0.0010 & 0.0000\\
0.0000 & 0.0000
\end{tabular}
} & \multicolumn{1}{|l}{%
\begin{tabular}
[c]{ll}%
$F_{n}(1)$ & $\tilde{F}_{n}(1)$\\
0.9257 & 0.9257\\
0.6509 & 0.0463\\
0.2891 & 0.0002\\
0.0730 & 0.0000\\
0.0099 & 0.0000\\
0.0007 & 0.0000
\end{tabular}
}\\\hline
\end{tabular}
\
\]

\end{center}

Table 4.2 is the equivalent of Table 4.1 with a lower volatility $\sigma
=0.12$. We notice that $F_{n}(1)$ and $\tilde{F}_{n}(1)$ decrease as $\sigma$
decreases. We also have an interesting observation that the trend of
$\tilde{F}_{2}(1)$ is not monotone in $\mu$. Again, this is because the
occurrence of $\tilde{\tau}_{a}^{n}$ for $n\geq2$ necessitates a recovery for
the previous running maximum. Smaller drift does imply higher drawdown risk,
meanwhile the recovery becomes more difficult.

\section{Insurance of frequent relative drawdowns}

In this section, we consider insurance policies protecting against the risk of
frequent drawdowns. We denote the price of an underlying asset by
$S=\{S_{t},t\geq0\}$, with dynamics
\[
\mathrm{d}S_{t}=rS_{t}\mathrm{d}t+\sigma S_{t}\mathrm{d}W_{t}^{%
\mathbb{Q}
}\text{,\qquad}S_{0}=s_{0}\text{,}%
\]
where $r>0$ is the risk-free rate, $\sigma>0$ and $\{W_{t}^{%
\mathbb{Q}
},t\geq0\}$ is a standard Brownian motion under a risk-neutral measure $%
\mathbb{Q}
$. It is well known that
\begin{equation}
S_{t}=s_{0}\mathrm{e}^{X_{t}}, \label{SX}%
\end{equation}
where $X_{t}=(r-\frac{1}{2}\sigma^{2})t+\sigma W_{t}^{%
\mathbb{Q}
}$.

In practice, drawdowns are often quoted in percentage. For fixed $0<\alpha<1$,
we denote the time of the first relative drawdown over size $\alpha$ by
\[
\eta_{\alpha}(S)=\inf\left\{  t\geq0:M_{t}^{S}-S_{t}\geq\alpha M_{t}%
^{S}\right\}  ,
\]
where $M_{t}^{S}=\sup_{0\leq u\leq t}S_{u}$ represents the running maximum of
$S$ by time $t$. By (\ref{SX}), it is easy to see that the relative drawdown
of the geometric Brownian motion $S$ corresponds to the actual drawdown of a
drifted Brownian motion $X$, namely%
\[
\eta_{\alpha}(S)=\inf\left\{  t\geq0:M_{t}^{X}-X_{t}\geq-\log(1-\alpha
)\right\}  =\tau_{\bar{\alpha}}(X),
\]
where $\bar{\alpha}=-\log(1-\alpha)$. Similarly, we denote the relative
drawdown times with and without recovery by%
\[
\tilde{\eta}_{\alpha}^{n}(S)=\inf\{t>\tilde{\eta}_{\alpha}^{n-1}(S):M_{t}%
^{S}-S_{t}\geq\alpha M_{t}^{S},M_{t}^{S}>M_{\tilde{\eta}_{\alpha}^{n-1}%
(S)}^{S}\},
\]
and
\[
\eta_{\alpha}^{n}(S)=\inf\{t>\eta_{\alpha}^{n-1}(S):M_{[\eta_{\alpha}%
^{n-1}(S),t]}^{S}-S_{t}\geq\alpha M_{[\eta_{\alpha}^{n-1}(S),t]}^{S}\}\text{,}%
\]
respectively. Therefore, we have
\begin{equation}
\tilde{\eta}_{\alpha}^{n}(S)=\tilde{\tau}_{\bar{\alpha}}^{n}(X)\qquad
\text{and\qquad}\eta_{\alpha}^{n}(S)=\tau_{\bar{\alpha}}^{n}(X). \label{it}%
\end{equation}

Next, we consider two types of insurance policies offering a protection
against relative drawdowns. For the first one, we assume that the seller pays
the buyer $\$k$ at time $T$ if a total of $k$ relative drawdowns over size
$0<\alpha<1$ occurred prior to time $T$ (for all $k$). For the relative
drawdown times with and without recovery, by (\ref{it}), the risk-neutral
prices are given by
\[
\tilde{V}_{1}(T)=\mathrm{e}^{-rT}\sum_{k=1}^{\infty}k%
\mathbb{Q}
\left\{  \tilde{N}_{T}^{\bar{\alpha}}(X)=k\right\}  =\mathrm{e}^{-rT}%
\mathbb{E}^{%
\mathbb{Q}
}[\tilde{N}_{T}^{\bar{\alpha}}(X)],
\]
and
\[
V_{1}(T)=\mathrm{e}^{-rT}\sum_{k=1}^{\infty}k%
\mathbb{Q}
\left\{  N_{T}^{\bar{\alpha}}(X)=k\right\}  =\mathrm{e}^{-rT}\mathbb{E}^{%
\mathbb{Q}
}[N_{T}^{\bar{\alpha}}(X)],
\]
respectively. For the second type of policies, the seller pays the buyer $\$1$
at the time of each relative drawdown time as long as it occurs before
maturity $T$. Hence, their risk-neutral prices are
\[
\tilde{V}_{2}(T)=\sum_{k=1}^{\infty}\mathbb{E}^{%
\mathbb{Q}
}[\mathrm{e}^{-r\tilde{\tau}_{\bar{\alpha}}^{k}(X)};\tilde{\tau}_{\bar{\alpha
}}^{k}(X)\leq T],
\]
and
\[
V_{2}(T)=\sum_{k=1}^{\infty}\mathbb{E}^{%
\mathbb{Q}
}[\mathrm{e}^{-r\tau_{\bar{\alpha}}^{k}(X)};\tau_{\bar{\alpha}}^{k}(X)\leq
T]\text{,}%
\]
respectively.

\begin{corollary}
\label{price} For $\lambda>0$, we have%
\[%
\begin{array}
[c]{ll}%
\int_{0}^{\infty}\mathrm{e}^{-\lambda T}V_{1}(T)\mathrm{d}T=\frac{1}%
{\lambda+r}\frac{\bar{c}_{\lambda+r}/\bar{b}_{\lambda+r}}{1-\bar{c}%
_{\lambda+r}/\bar{b}_{\lambda+r}}, & \int_{0}^{\infty}\mathrm{e}^{-\lambda
T}\tilde{V}_{1}(T)\mathrm{d}T=\frac{1}{\lambda+r}\frac{\bar{c}_{\lambda
+r}/\bar{b}_{\lambda+r}}{1-\mathrm{e}^{-\bar{\beta}_{\lambda+r}^{+}a}\bar
{c}_{\lambda+r}/\bar{b}_{\lambda+r}},\\
\int_{0}^{\infty}\mathrm{e}^{-\lambda T}V_{2}(T)\mathrm{d}T=\frac{1}{\lambda
}\frac{\bar{c}_{\lambda+r}/\bar{b}_{\lambda+r}}{1-\bar{c}_{\lambda+r}/\bar
{b}_{\lambda+r}}, & \int_{0}^{\infty}\mathrm{e}^{-\lambda T}\tilde{V}%
_{2}(T)\mathrm{d}T=\frac{1}{\lambda}\frac{\bar{c}_{\lambda+r}/\bar{b}%
_{\lambda+r}}{1-\mathrm{e}^{-\bar{\beta}_{\lambda+r}^{+}a}\bar{c}_{\lambda
+r}/\bar{b}_{\lambda+r}},
\end{array}
\]
where $\bar{b}_{\lambda}=\frac{\bar{\beta}_{\lambda}^{+}\mathrm{e}%
^{-\bar{\beta}_{\lambda}^{-}\bar{\alpha}}-\bar{\beta}_{\lambda}^{-}%
\mathrm{e}^{-\bar{\beta}_{\lambda}^{+}\bar{\alpha}}}{\mathrm{e}^{-\bar{\beta
}_{\lambda}^{-}\bar{\alpha}}-\mathrm{e}^{-\bar{\beta}_{\lambda}^{+}\bar
{\alpha}}}$, $\bar{c}_{\lambda}=\frac{\bar{\beta}_{\lambda}^{+}-\bar{\beta
}_{\lambda}^{-}}{\mathrm{e}^{-\bar{\beta}_{\lambda}^{-}\bar{\alpha}%
}-\mathrm{e}^{-\bar{\beta}_{\lambda}^{+}\bar{\alpha}}}$ and $\bar{\beta
}_{\lambda}^{\pm}=\frac{-r+\frac{1}{2}\sigma^{2}\pm\sqrt{(r-\frac{1}{2}%
\sigma^{2})^{2}+2\lambda\sigma^{2}}}{\sigma^{2}}$.
\end{corollary}

\begin{proof}
We provide the proof for $\int_{0}^{\infty}V_{1}(T)\mathrm{e}^{-\lambda
T}\mathrm{d}T$ and $\int_{0}^{\infty}V_{2}(T)\mathrm{e}^{-\lambda T}%
\mathrm{d}T$ only. The other two results can be derived in a similar fashion.
From the definition of $N_{T}^{\bar{\alpha}}(X)$, we have the following
relation
\[
\mathbb{E}^{%
\mathbb{Q}
}\left[  N_{T}^{\bar{\alpha}}(X)\right]  =\sum_{k=1}^{\infty}\mathbb{Q}%
\left\{  N_{T}^{\bar{\alpha}}(X)\geq k\right\}  =\sum_{k=1}^{\infty}%
\mathbb{Q}\left\{  \tau_{\bar{\alpha}}^{k}(X)\leq T\right\}  .
\]
By (\ref{L tau}), it follows that%
\begin{align*}
\int_{0}^{\infty}V_{1}(T)\mathrm{e}^{-\lambda T}\mathrm{d}T  &  =\int
_{0}^{\infty}\mathrm{e}^{-(\lambda+r)T}\mathbb{E}^{%
\mathbb{Q}
}[N_{T}^{\bar{\alpha}}(X)]\mathrm{d}T\\
&  =\sum_{k=1}^{\infty}\int_{0}^{\infty}\mathrm{e}^{-(\lambda+r)T}%
\mathbb{Q}\left\{  \tau_{\bar{\alpha}}^{k}(X)\leq T\right\}  \mathrm{d}T\\
&  =\frac{1}{\lambda+r}\sum_{k=1}^{\infty}\mathbb{E}^{%
\mathbb{Q}
}[\mathrm{e}^{-(\lambda+r)\tau_{\bar{\alpha}}^{k}(X)}]\\
&  =\frac{1}{\lambda+r}\sum_{k=1}^{\infty}\left(  \frac{\bar{c}_{\lambda+r}%
}{\bar{b}_{\lambda+r}}\right)  ^{n}\\
&  =\frac{1}{\lambda+r}\frac{\bar{c}_{\lambda+r}/\bar{b}_{\lambda+r}}%
{1-\bar{c}_{\lambda+r}/\bar{b}_{\lambda+r}}.
\end{align*}
For $\int_{0}^{\infty}V_{2}(T)\mathrm{e}^{-\lambda T}\mathrm{d}T$, by Fubini's
theorem and (\ref{L tau}), we have%
\begin{align*}
\int_{0}^{\infty}V_{2}(T)\mathrm{e}^{-\lambda T}\mathrm{d}T  &  =\sum
_{k=1}^{\infty}\int_{0}^{\infty}\mathbb{E}^{%
\mathbb{Q}
}[\mathrm{e}^{-r\tau_{\bar{\alpha}}^{k}(X)};\tau_{\bar{\alpha}}^{k}(X)\leq
T]\mathrm{e}^{-\lambda T}\mathrm{d}T\\
&  =\sum_{k=1}^{\infty}\int_{0}^{\infty}\int_{0}^{T}\mathrm{e}^{-rt}%
\mathbb{Q}
\left\{  \tau_{\bar{\alpha}}^{k}(X)\in\mathrm{d}t\right\}  \mathrm{e}%
^{-\lambda T}\mathrm{d}T\\
&  =\sum_{k=1}^{\infty}\frac{1}{\lambda}\int_{0}^{\infty}\mathrm{e}%
^{-(\lambda+r)t}%
\mathbb{Q}
\left\{  \tau_{\bar{\alpha}}^{n}(X)\in\mathrm{d}t\right\} \\
&  =\sum_{k=1}^{\infty}\frac{1}{\lambda}\left(  \frac{\bar{c}_{\lambda+r}%
}{\bar{b}_{\lambda+r}}\right)  ^{n}\\
&  =\frac{1}{\lambda}\frac{\bar{c}_{\lambda+r}/\bar{b}_{\lambda+r}}{1-\bar
{c}_{\lambda+r}/\bar{b}_{\lambda+r}}.
\end{align*}
This completes the proof.
\end{proof}

\begin{remark}
It is worth pointing out that, through expansion of the randomized prices in
Corollary \ref{price} in terms of exponentials, it is possible to obtain
semi-static hedging portfolios as in \cite{CarrZhanHaji}. Moreover, capped
insurance contracts against frequency of drawdowns can also be formulated and
priced using Theorems \ref{thm M til L}, \ref{yyz}, and Corollary \ref{aab}.
\end{remark}

To conclude, we consider a pricing example for the four types of insurance
contracts proposed earlier. The same numerical Laplace transform approach as
in the last section is applied.

\begin{center}
\textbf{Table 5.1 }Insurance contracts prices when $\alpha=15\%$ and $r=5\%$%
\[%
\begin{tabular}
[c]{ll|l|l|l|l}\hline
&  & $V_{1}(T)$ & $\tilde{V}_{1}(T)$ & $V_{2}(T)$ & $\tilde{V}_{2}(T)$\\\hline
$T=1$ & $\sigma=0.1$ & 0.1102 & 0.1091 & 0.1120 & 0.1108\\
$T=2$ & $\sigma=0.1$ & 0.3011 & 0.2769 & 0.3131 & 0.2885\\
$T=3$ & $\sigma=0.1$ & 0.4743 & 0.4031 & 0.5058 & 0.4318\\\hline
$T=1$ & $\sigma=0.2$ & 1.1777 & 0.7873 & 1.2043 & 0.8081\\
$T=2$ & $\sigma=0.2$ & 2.3815 & 1.1842 & 2.4977 & 1.2550\\
$T=3$ & $\sigma=0.2$ & 3.4651 & 1.4519 & 3.7279 & 1.5890\\\hline
\end{tabular}
\
\]

\end{center}

As expected, Table 5.1 shows that type 2 contracts have higher prices than
type 1 contracts because of earlier payments (at the moment of each drawdown
time instead of the maturity $T$). It also shows that $\tilde{V}_{1}(T)$ and
$\tilde{V}_{2}(T)$ are respectively lower than $V_{1}(T)$ and $V_{2}(T)$ due
to $\tau_{a}^{n}\leq\tilde{\tau}_{a}^{n}$. All the prices increase as $T$
increases or $\sigma$ increases. Moreover, we can expect that the prices will
decrease as $\alpha$ or $r$ increases. The latter is due to a higher discount
rate which is the risk-free rate under the risk-neutral measure $%
\mathbb{Q}
$.

\vskip0.5cm

\noindent\textbf{Acknowledgments.} The authors would like to thank Professor
Gord Willmot and an anonymous referee for their helpful remarks and
suggestions. Support for David Landriault from a grant from the Natural
Sciences and Engineering Research Council of Canada is gratefully
acknowledged, as is support for Bin Li from a start-up grant from the
University of Waterloo.

\baselineskip13pt


\end{document}